\documentclass{ifacconf}

\usepackage{color}
\usepackage{booktabs} 
  \usepackage{amsmath}
\usepackage{amssymb}
\usepackage{amsbsy}
\usepackage{algorithm}               
\usepackage{algpseudocode}
\usepackage{graphicx}  
\usepackage{psfrag}
\usepackage{diagbox}
\usepackage{bm}
\usepackage{bbm}                                                       
\allowdisplaybreaks

\usepackage{bm}
\usepackage{bbm}
    
\RequirePackage{epsfig}
\RequirePackage{psfrag}
\RequirePackage{pstool}
\RequirePackage{epstopdf}
\RequirePackage{graphicx}
\usepackage{bbm}
\usepackage{subfigure}
\RequirePackage{mathrsfs}

\RequirePackage{amsmath}
\RequirePackage{amssymb}

\RequirePackage{amsbsy}
\RequirePackage{amsfonts}
\RequirePackage{mathrsfs}
\usepackage{algorithm}   
\usepackage{algpseudocode}

\usepackage{psfrag} 
\usepackage{amsmath} 

\newtheorem{proposition}{\bf Proposition}
\newtheorem{definition}{Definition}
\newtheorem{lemma}{Lemma}
\newtheorem{theorem}{Theorem}
\newtheorem{remark}{Remark}

\newcommand{\R}{\mathbb{R}}

\newcommand{\V}{\mathbb{V}}
\newcommand{\ran}{\mathrm{ran}}
\newcommand{\mbf}[1]{\mathbf{#1}}
\newcommand{\mc}[1]{\mathcal{#1}}
\newcommand{\tit}[1]{\textit{#1}}
\newcommand{\mrm}[1]{\mathrm{#1}}
\newcommand{\delx}{\frac{\partial}{\partial \mathbf{x}}}

\usepackage{algorithm}   
\usepackage{algpseudocode}

\DeclareMathOperator*{\argmin}{arg\,min}

\usepackage{enumitem}
\usepackage{tabularx,booktabs,ragged2e}

\usepackage{graphicx}      
\usepackage{natbib}        

\begin{document}
\begin{frontmatter}

\title{Using Seminorms To Analyze Contraction of Switched Systems With Only Non-Contracting Modes } 

\author[First]{Edwin Baum}
\author[First]{Zonglin Liu} 
\author[Second]{Yuzhen Qin} 
\author[First]{Olaf Stursberg} 

\address[First]{Control and System Theory, Dept. of Electrical Engineering and Computer Science, University of Kassel, Germany.\\
	Email:  uk069490@student.uni-kassel.de, z.Liu@uni-kassel.de,  y.qin@donders.ru.nl,  stursberg@uni-kassel.de.}
\address[Second]{Department of Machine Learning and Neural Computing, Donders Institute for Brain, Cognition and Behaviour, Radboud University, the Netherlands\\
 yuzhen.qin@donders.ru.nl}

\begin{abstract}
This paper investigates contraction properties of switched dynamical systems for the case that all modes are non-contracting, thereby extending existing results that require at least one mode to be contracting.  Leveraging the property that unstable systems may still exhibit stable behavior  within certain subspaces, conditions are provided which ensure contracting evolution within a given subspace of the state space of the switched system. These conditions are derived using the concepts of seminorms and semi-contracting systems. Then, by selecting a set of subspaces whose corresponding seminorms form a separating family of the state space, and by verifying whether a given mode is contracting in each subspace, conditions on the activation time of each mode are provided by which contraction on the complete state space is guaranteed. Numerical examples are presented for illustration.
\end{abstract}
	

\end{frontmatter}

\section{Introduction}
\label{sec:intro}

Contractive dynamical systems possess the property that any two trajectories converge to each other regardless of their initial conditions. This property proves useful for a wide range of tasks, such as the stability analysis of equilibria \citep{sontag2010contractive}, the development of asymptotic tracking controllers \citep{miljkovic2025reference}, and the evaluation of synchronization conditions in networked systems \citep{jafarpour_weak_2022, liu2025synchronization}.
For systems that can switch dynamics according to an external switching signal, such as those considered in \citep{hespanha1999stability, liu2018optimizing}, both the local dynamics of each mode and the switching signal influence whether the overall evolution is contracting.

Unlike the numerous studies on the property of contraction of time-triggered or state-triggered switching systems, such as those in \citep{di2014incremental, di2013incremental, burden2018contraction}, there are only a few publications addressing switched systems. Among these, the studies in \citep{yang2016incremental, mccloy2021contraction, di2014contraction} have shown that even if all modes are contracting, only switching signals that satisfy certain dwell time requirements can ensure a contracting evolution. For switched systems with a mix of contracting and non-contracting modes, the recent work \citep{yin_contraction_2023} generalizes previous results by providing an upper bound on the dwell time of non-contracting modes, and a lower bound on the dwell time of contracting modes to guarantee contraction. To the best of the authors' knowledge, no work has considered the case when all modes are non-contracting. (The authors in \citep{yin_contraction_2023} demonstrate the potential to extend their results to this case though, and the precise distinction of that work from the present paper is described in Sec.~2.)

Motivated by the subspace decomposition method developed in \cite{qin_synchronization_2020} for demonstrating synchronization of networked systems with switching unconnected graphs, this paper exploits the concepts of seminorms and semi-contraction systems, as discussed in \cite{cisneros-velarde_contraction_2022, jafarpour_weak_2022, de_pasquale_dual_2024} in order to show the contraction of switching non-contracting modes. Based on these concepts, the first question addressed is: under what switching time conditions does the switched system contract within a given subspace, even if it is non-contracting in the full state space? Building on this result, a set of subspaces is selected to jointly determine the switching time conditions under which the switched system contracts in the full state space. In Sec.~2,  the considered class of switched systems and the contraction problem are  described  in detail, along with relevant mathematical preliminaries. 
In Sec.~3, the rules for selecting subspaces and the switching time conditions that enable contraction within the subspaces and the full state space are introduced.  Two numerical examples confirming the results are provided in Sec.~4, followed by conclusions and an outlook on future directions in Sec.~5.

\section{Problem Description}
\label{sec:prob}

This paper considers a  type of switched systems in the form of:
\begin{align}\label{eq:sys}
    \dot{x}(t) =   \mathbf{f}_{\sigma (t)}(x(t)),\quad  x(t_0) =x_0
\end{align}
with the system state $x \in \mathbb{R}^n$, the initial time $t_0$, and the initial state $x_0$. Given a set  $\mathcal{M} = \{1,\ldots, M\}$ of system modes,  the switching signal $\sigma(t): [t_0, \infty) \mapsto \mathcal{M}$ is assumed to be a piecewise constant, right-continuous function that determines the active mode  at any time $t \ge t_0$. The switching instants are defined by $\mathcal{T}_\sigma:= \{t_0, \ldots, t_k,\ldots\} \subseteq [t_0, \infty)$ with  the dwell time $\tau_k := t_{k+1} -t_k >0$.
\begin{definition}\label{UfC}(\textbf{Uniformly Contracting Switched systems}(\cite{yin_contraction_2023}))
The switched system \eqref{eq:sys} for a given switching signal $\sigma(t)$  is 
 called uniformly contracting if there exist positive constants $c$ and $\alpha$, such that for any two
 different initial states $x_a(t_0)$ and $x_b(t_0)$ it holds for all $t \ge t_0$ that:
    \begin{align}\label{eq:uniformcontracting}
        \|x_a(t)-x_b(t)\|\leq \alpha \mrm{e}^{-c(t-t_0)} \|x_a(t_0)-x_b(t_0)\|.
    \end{align}
\end{definition}
To study the property of contraction of \eqref{eq:sys}, a variational system with new variables $y(t) \in \R^n$ is defined, where:
    \begin{align}\label{eq:var_sys}
         \dot{y}(t) = \mathbf{A}_{\sigma(t)}(x(t)) y(t),~ \mathbf{A}_{\sigma(t)}(x(t))= \delx \mathbf{f}_{\sigma(t)} (x(t)).
     \end{align}
\begin{theorem}\label{theorem:contractingvariation}(\cite{yin_contraction_2023})
The switched system \eqref{eq:sys} is uniformly contracting for a given switching signal $\sigma(t)$, if and only if the 
  variational system \eqref{eq:var_sys} is  exponentially stable for all  $y(t_0) \in \R^n$ and $t \ge t_0$ if:
    \begin{align}\label{eq:uniformstabilizing}
        \|y(t)\|\leq \alpha \mrm{e}^{-c(t-t_0)} \|y(t_0)\|
    \end{align}
with positive constants $c$ and for $\alpha$ according to Def.~\ref{UfC}. 
\end{theorem}

Based on this theorem, the work by \cite{yin_contraction_2023}  established conditions for contraction of the switched system \eqref{eq:sys} if only a few modes in $\mathcal{M}$ are contracting. This represents an advancement over previous results by \cite{yang2016incremental, mccloy2021contraction}, in which all modes were required to be contractive. The authors there assign a  set of piecewise Lyapunov  functions to each system mode $q \in \mathcal{M}$. They show that if  certain minimal dwell time requirements are satisfied  by all contracting modes of the switching signal  $\sigma(t)$, while certain maximal  dwell time requirements are met for the remaining non-contracting modes, the value of the Lyapunov  function will   decrease exponentially over time, thereby implying \eqref{eq:uniformstabilizing}. They also show that even if all modes are non-contracting, the  switched system \eqref{eq:sys} may still be contracting when the value reduction caused by mode switching (i.e. the switching of  Lyapunov function) compensates for the increase induced by the flow dynamics of each non-contracting mode. However, the scheme that relies solely on mode switching to provide contraction can be overly conservative in practice, as it requires the modes of  $\sigma(t)$  to switch at a high frequency. 

To relax this conservative  condition to the case that all modes are non-contracting, the concept of semi-contracting systems is introduced in the following section. It is shown that even for non-contracting modes, a Lyapunov function that is defined in a certain subspace of $ \mathbb{R}^n$ may still   decrease over time. In this way, the contraction of the switched system \eqref{eq:sys} can already be ensured if the selected subspaces exhibit contracting dynamics for some mode of  $\sigma(t)$. Before presenting the main result, some important mathematical preliminaries are introduced.

\subsection{Mathematical Preliminaries}
In  this paper, the notion $\mrm{ran}(\cdot)$ denotes the range of a matrix, $\mrm{span}(\cdot)$ denotes the linear span of a set of vectors, and $\mrm{ker}(\cdot)$ denotes the kernel of a matrix.
\begin{definition}\label{def:orth_proj}(\textbf{Orthogonal Projection Matrices} (\cite{finite-dimensional}))
For a subspace $\V \subseteq \R^n$, a symmetric matrix  $\mbf{\Pi}_\V$ is called an orthogonal projection matrix onto $\V$ if:
     \begin{align}\label{eq:proj_def}
\mbf{\Pi}_\V^2 = \mbf{\Pi}_\V, ~ \mrm{ran}(\mbf{\Pi}_\V ) = \V.
     \end{align}
  \end{definition}   
Let $\{\mbf{v_1}, \mbf{v_2},\ldots,\mbf{v_h}\}\subset \V$, $h \le n$, be an orthonormal list, such that:
     \begin{align}
         \V = \mrm{span}(\{\mbf{v_1}, \mbf{v_2},\ldots,\mbf{v_h}\}).
     \end{align}
Define the matrix $\mbf{V} := [\mbf{v_1}, \mbf{v_2},...,\mbf{v_h}] \in \R^{n\times h}$, an orthogonal projection matrix $\mbf{\Pi}_\V$ onto $\V$ can be determined by:
     \begin{align}\label{eq:constractionPI}
         \mbf{\Pi}_\V =\mbf{{V} {V}}^T.
     \end{align}
For the orthogonal complement subspace $\V^\perp \subseteq \R^n$,   let $\{\mbf{u_1,u_2,\ldots,u_{n-h}}\}\subset \V^\perp$ be the orthonormal list, such that:
\begin{align}
         \V^\perp = \mrm{span}(\{\mbf{u_1}, \mbf{u_2},...,\mbf{u_{n-h}}\}).
     \end{align}
   Then, the matrix $\mbf{V}$ together with the matrix   $\mbf{U} := [\mbf{u_1}, \mbf{u_2},...,\mbf{u_{n-h}}] \in \R^{n\times (n-h)}$ constitute a  nonsingular and orthogonal matrix:
    \begin{align}\label{eq:Tmatrix}
   \mbf{T} = [\mbf{V}| \mbf{U}] \in \R^{n\times n}
        \end{align}
        which  diagonalizes $\mbf{\Pi}_\V$ by:
 \begin{align}\label{eq:proj_decom}
     \mbf{T}^T \mbf{\Pi}_\V  \mbf{T} = \begin{bmatrix}
         \mbf{I}_h & \mbf{0}_{n\times (n-h)} \\
         \mbf{0}_{(n-h)\times p} & \mbf{0}_{(n-h)\times (n-h)}
     \end{bmatrix}.
 \end{align}
Based on \eqref{eq:proj_decom}, there also exists:
 \begin{align}\label{eq:orth_proj_constr}
     \mbf{\Pi}_\V = \mbf{I}_n - \mbf{\Pi}_{\V^\perp}.
 \end{align}
 \begin{definition}(\textbf{Invariant subspace} (\cite{Bullo_Con_Theo_Dynam_Sys}))\label{def:invariantspace}
     Given a matrix $\mbf{A} \in \R^{n\times n}$ a subspace $\V \subseteq \R^n$ is called an invariant subspace of  $\mbf{A}$, if for any vector $\mbf{v} \in \V$, there  exists $\mbf{A} \mbf{v} \in \V$.
 \end{definition}
Based on Def.~\ref{def:invariantspace}, given  a subspace $\V \subseteq \R^{n}$ and a  symmetric matrix $\mbf{P} \in \R^{n\times n}$ satisfying $\ker(\mbf{P}) = \V^\perp$, the relation $\mbf{P} \V \subseteq \V$ must hold, which implies that  $\V$ is  an invariant subspace of $\mbf{P}$. In addition, for  the matrix $\mbf{T}$ in \eqref{eq:Tmatrix},  there also exists:
 \begin{align}\label{eq:P_tilde}
\mbf{T}^T \mbf{P} \mbf{T} \hspace{-0.5mm}= \hspace{-0.5mm}\begin{bmatrix}
\mbf{V}^T \mbf{P}\mbf{V} & \mbf{V}^T \mbf{P} \mbf{U}\\
\mbf{U}^T \mbf{P} \mbf{V} & \mbf{U}^T \mbf{P}\mbf{U}
\end{bmatrix} 
\hspace{-0.5mm}= \hspace{-0.5mm}\begin{bmatrix}
    \Tilde{\mbf{P}} & \mbf{0}\\
    \mbf{0}& \mbf{0}
\end{bmatrix},~\Tilde{\mbf{P}} \in \R^{h\times h}
 \end{align}
due to $\ran(\mbf{P}) = (\ker(\mbf{P}))^\perp$ and 
$\mbf{U}^T \mbf{P} \mbf{V} = \mbf{0}$.

\begin{definition}\label{def:semnorm}(\textbf{Seminorms} (\cite{rudin_functional_1996}))
A function $|||\cdot |||: \R ^{n} \rightarrow \R_{\geq 0}$  is a \tit{seminorm}, if:
\begin{enumerate}
    \item $||| \alpha \mbf{v}||| = |\alpha|\, |||\mbf{v} ||| $ for all $\mbf{v} \in \R^{n}$ and $\alpha \in \R$;
    \item $|||\mbf{v}+\mbf{w}|||\leq |||\mbf{v}|||+|||\mbf{w}|||$ for all $\mbf{v},\mbf{w} \in \R^n$.
\end{enumerate}
The kernel of a seminorm  is a subspace of $\R^n$  defined as:
\begin{align*}
    \mrm{ker}(|||\cdot |||):= \{ \mbf{v}\in \R^n ~|~ ||| \mbf{v}||| = 0\}.
\end{align*}
If $ \mrm{ker}(|||\cdot |||) = \{\mbf{0}\}$, then $|||\cdot |||$ is a \tit{norm}  denoted by $\| \cdot \|$. \\
\end{definition}

\begin{definition}\label{def:sep_fam}(\textbf{Separating Family} (\cite{rudin_functional_1996})) Let $\mc{F} = \{|||\cdot |||_i\}_{i \in \mc{I}}$, $\mc{I} := \{1,2,...,I\},\, I \in \mathbb{N}$, be a family of seminorms on $\R^n$. $\mc{F}$ is called  a \tit{Separating Family} of $\R^n$, if $\mbf{v} = \mbf{0}$ is the only vector in
 $\R^n$ satisfying:
\begin{align*}
  |||\mbf{v} |||_i = 0, ~ \forall i\in \mc{I}.
\end{align*}
\end{definition}
Note that  a  separating family also induces a norm $\| \cdot \|_{\mc{F}}$ on $\R^n$ due to the fact that $\ker(\sum_{i \in \mc{I}} |||\cdot|||_i) = \{\mbf{0}\}$.

\begin{proposition}\label{prop3}(\cite{treves2016topological})
 Let $\mc{F} = \{|||\cdot |||_i\}_{i \in \mc{I}}$, $\mc{I} := \{1,2,...,I\},\, I \in \mathbb{N}$, be a family of seminorms on $\R^n$ and if:
    \begin{align}\label{eq:conditionseperating}
        \bigcap_{i\in \mc{I}} \ker(|||\cdot|||_i) = \{\mbf{0}\}
    \end{align}
    holds,  $\mc{F}$ must be a separating family of $\R^n$.
\end{proposition}
A seminorm can be induced from, e.g., the Euclidean norm $\|\cdot \|_2$. For instance, given a subspace  $\V \subseteq \R^{n}$ and the corresponding orthogonal projection $\mbf{\Pi}_\V$, the Euclidean norm $\|\mbf{\Pi}_\V \,(\hspace{.1cm}\cdot\hspace{.1cm})\,\|_2$ defines a seminorm $|||\cdot|||$ on $\R^n$ by:
\begin{align}\label{eq:Proj_sem_norm}
    |||\cdot||| :=\|\mbf{\Pi}_\V \,(\hspace{.1cm}\cdot\hspace{.1cm})\,\|_2
\end{align}
with $\mrm{ker}(|||\cdot |||)=\V^\perp$.
Furthermore,  a weighted seminorm with a semi-positive definite matrix $\mbf{P}\in \R^{n\times n}$, $\ker(\mbf{P}) = \V^\perp$, is defined as:
\begin{align}\label{eq:P_sem_norm}
    |||\cdot |||_{2,\mbf{P}^{1/2}} = \|\mbf{P}^{\frac{1}{2}}\,(\hspace{.1cm}\cdot\hspace{.1cm})\,\|_2
\end{align}
with $\mrm{ker}( |||\cdot |||_{2,\mbf{P}^{1/2}})=\V^\perp$.
The  induced weighted matrix logarithmic seminorm of  $\mathbf{A}$ is given by:
\begin{align}\label{eq:l2-measure}
& \mu_{2,\mbf{P}^{1/2}}(\mathbf{A}):=\argmin\limits_{\mathbf{P}\mathbf{A}\boldsymbol{\Pi}_{\V^\perp}+ \boldsymbol{\Pi}_{\V^\perp}\mathbf{A}^T  \mathbf{P} \preceq 2b \mathbf{P} }b 
\end{align}
according to \cite{de_pasquale_dual_2024}.

\begin{definition}\label{def:semcon}
\textbf{(Semi-contracting Systems and Infinitesimal Invariance}  (\cite{jafarpour_weak_2022})).  
Given  a subspace  $\V \subseteq \R^{n}$ and let $ |||\cdot |||_{2,\mbf{P}^{1/2}}$ 
be a weighted  seminorm on $\mathbb{R}^n$ with a semi-positive definite matrix $\mbf{P}\in \R^{n\times n}$ satisfying $\ker(\mbf{P}) = \V^\perp$. 
Then, the mode  $\dot{x}=\mathbf{f}_q(x)$, $q \in \mathcal{M}$ of \eqref{eq:sys}  is  infinitesimally semi-contracting with respect to   $ |||\cdot |||_{2,\mbf{P}^{1/2}}$   with rate $c > 0$ if:
\begin{align}    \label{eq:semicontracting}
     \mu_{|||\cdot |||_{2,\mbf{P}^{1/2}} } \!\big( \delx\mathbf{f}_q( x) \big) \;\le\; -c,
    \quad \forall ~ x \in \mathbb{R}^n.
\end{align}
In addition, the subspace  $\V$ is called infinitesimally-invariant with respect to $\dot{x}=\mathbf{f}_q(x)$  if:
\begin{align}    \label{eq:inv-diff}
    \delx\mbf{f}_q(x)\,\V =\mbf{A}_q(x)\,\V  &\subseteq \V
\end{align}
or equivalently:
\begin{align}   \label{eq:projection-form}
    \mbf{\Pi}_{\V^\perp} \delx\mbf{f}_q(x) \mbf{\Pi}_{\V} = \mbf{\Pi}_{\V^\perp}\mbf{A}_q(x) \mbf{\Pi}_{\V} = \mbf{0}_{n \times n}
\end{align}
hold for all $ x \in \mathbb{R}^n$.

\end{definition}

\section{Contraction Analysis  Using a Separating Family of Seminorms}\label{sec:Contraction}

For  the switched system \eqref{eq:sys} and 
any finite time interval $[t_a, t_b) \subseteq [t_0, \infty)$, let $N_{q}(t_a, t_b) \in \mathbb{N}$  denote the  number of times  mode $q$ is activated, and  let $T_q(t_a, t_b) \in \mathbb{R}_{\ge 0}$ denote its total activation time within the interval $[t_a, t_b)$. Then, the concepts of average dwell time and average leave time are introduced:
\begin{definition}
(\noindent\textbf{Mode-Dependent Average Dwell and Leave Time} (\cite{yin_contraction_2023}))  
For a switching signal $\sigma : [t_0, \infty) \to \mathcal{M}$ and any  mode $q \in \mathcal{M}$ of \eqref{eq:sys}, the positive constants $\underline{\tau}_{q}$ and $\overline{\tau}_{q}$ are called the  mode-dependent average dwell time (MDADT) and the  mode-dependent average leave   time  (MDALT) of $q$, respectively, if there exist positive constants  $\underline{N}_{q}$ and  $\overline{N}_{q}$, such that for any finite time interval  $[t_a, t_b) \subseteq [t_0, \infty)$:
\begin{align}
    N_{q}(t_a, t_b) \leq  \underline{N}_{q} + \frac{T_q(t_a, t_b)}{\underline{\tau}_{q}},
    \label{eq:MDADT}\\ 
     N_{q}(t_a, t_b) \geq \overline{N}_{q} + \frac{T_q(t_a, t_b)}{\overline{\tau}_{q}}.
    \label{eq:MDALT}
\end{align}
\end{definition}


\subsection{Contraction Analysis for a Subspace $\V \subseteq \R^n$}

Next, given  a subspace $\V \subseteq \R^n$ and a  semi-positive definite matrix $\mbf{P}\in \R^{n\times n}$ satisfying $\ker(\mbf{P}) = \V^\perp$,
 the  evolution of the mode $q \in \mathcal{M}$ of the variational system \eqref{eq:var_sys} is assumed to be measured by the weighted logarithmic seminorm $\mu_{2,\mbf{P}^{1/2}}(\mbf{A}_q(x(t)))$. Additionally, for a 
 bounded set $ \mc{D} \subseteq \R^{n}$, the system modes in $\mathcal{M}$ can be decomposed into two subsets, $\mathcal{M}=\mc{S}_{\mbf{P}} \cup \mc{U}_{\mbf{P}}$, where:
\begin{align}\label{eq:stablemodes}
\mc{S}_{\mbf{P}}:=\{q \in \mathcal{M}:~ \mu_{2,\mbf{P}^{1/2}}(\mbf{A}_q(x))<0, \, \forall x \in \mc{D}\},
\end{align}
and $\mc{U}_{\mbf{P}}$ collects the remaining modes.


\begin{lemma}\label{lem:sem-con}
For the switched system \eqref{eq:sys} with  a  switching signal $\sigma(t)$ and the switching times $\mathcal{T}_\sigma := \{ t_0, t_1, \ldots \}$, assume that  $x(t) \in \mc{D}  \subseteq \R^{n}$ uniformly holds for all $t \ge t_0$. Then, for a  subspace $\V \subseteq \R^n$, let  $\mbf{\Pi}_\V\in \R^{n\times n}$ be an  orthogonal projection matrix onto  $\V$.
Assume that $\V^\perp$ is infinitesimally-invariant with respect to all modes   $q \in \mathcal{M}$ of   \eqref{eq:sys}.  Furthermore, consider a set of semi-positive definite matrices
 $\mbf{P}_q\in \R^{n\times n}$ with $\ker(\mbf{P}_q) = \V^\perp$ for each mode $q \in  \mathcal{M}$, and constants $\beta_\mc{S} >1$ and $\beta_\mc{U} \in (0,1)$, such that:
 \begin{align}\label{eq:LMI1}
   &  {\mbf{P}}_{\sigma(t^+_i)} \preceq \beta_\mc{S} {\mbf{P}}_{\sigma(t_i^-)}, ~ \text{if} ~ \sigma(t_i^-) \in \mc{S}_{\mbf{P}_{\sigma(t_i^-)}}, \notag\\
   &  {\mbf{P}}_{\sigma(t^+_i)} \preceq \beta_\mc{U} {\mbf{P}}_{\sigma(t_i^-)}, ~ \text{if} ~ \sigma(t_i^-) \in \mc{U}_{\mbf{P}_{\sigma(t_i^-)}}
 \end{align}
 hold for all switching time $t_i \in \mc{T}_\sigma$.
Then,  if  there exist positive constants $\underline{m}$, $\overline{m}$, $\eta_{\mc{U}}$, and $\eta_{\mc{S}}$, such that:
\begin{align}
 &   \underline{m} \mbf{\Pi}_{\V}\preceq \mbf{P}_q \preceq  \overline{m} \mbf{\Pi}_{\V}, ~\forall q \in  \mathcal{M} \label{eq:LMI2} \\
& \mbf{P}_q \mbf{A}_q(x)  \mbf{\Pi}_{\V} \hspace{-0.5mm}+\hspace{-0.5mm}\mbf{\Pi}_{\V} \mbf{A}^T _q(x)  \mbf{P}_q \hspace{-0.5mm}\preceq \hspace{-0.5mm}-2\eta_{\mc{S}} \mbf{P}_q,~ \text{if} ~ q \in \mc{S}_{\mbf{P}_q} \label{eq:LMI3} \\
& \mbf{P}_q \mbf{A}_q(x)  \mbf{\Pi}_{\V} \hspace{-0.5mm}+\hspace{-0.5mm}\mbf{\Pi}_{\V} \mbf{A}^T _q(x)  \mbf{P}_q \hspace{-0.5mm}\preceq \hspace{-0.5mm}2\eta_{\mc{U}} \mbf{P}_q, ~\text{if} ~ q \in \mc{U}_{\mbf{P}_q} \label{eq:LMI4}
\end{align}
are satisfied,  and if the  average dwell  and  leave   times  of any mode $q \in  \mathcal{M}$ in the switching signal $\sigma(t)$ satisfy:
 \begin{align} 
\underline{\tau}_{q} &> \frac{\ln \beta_\mc{S}}{2\eta_{\mc{S}}}, ~ \text{if}~ q \in  \mc{S}_{\mbf{P}_q}\label{SMADTi}\\
\overline{\tau}_{q} &<-\frac{\ln \beta_\mc{U}}{2\eta_{\mc{U}}}, ~ \text{if}~ q \in \mc{U}_{\mbf{P}_q}\label{UMADTi}
\end{align}
there must exist $k, \lambda>0$, such that:
\begin{align}\label{eq:Conv_Rate}
    \|\mbf{\Pi}_{\V} y(t)\|_2\leq k \mrm{e}^{-\lambda(t-t_0)} \|\mbf{\Pi}_{\V} y(t_0)\|_2
\end{align}
holds for the variational system \eqref{eq:var_sys} for all $t \ge t_0$.

\end{lemma}
\begin{proof}
Given the switching signal $\sigma$, the following quadratic auxiliary function is defined for the
variational system \eqref{eq:var_sys}:
\begin{align}\label{eq:lyapulikefunction}
    V(y(t),t) = y^T(t) {\mbf{P}}_{\sigma(t)} y(t).
\end{align}
Assume that $\sigma(t) =q$ holds for $t \in [t_{i-1},t_i)$  between two subsequent switching times in $\mathcal{T}_\sigma$, there exists:
\begin{align}\label{eq:timederivative}
    \dot{V} (y(t),t) = y^T(t)({\mbf{P}}_q \mbf{A}_q(x(t))+ \mbf{A}^T_q(x(t)){\mbf{P}}_q )y(t).
\end{align}
Based on \eqref{eq:orth_proj_constr}, it is known that:
\begin{align}\label{eq:Aqx}
   {\mbf{P}}_q\mbf{A}_q(x)= {\mbf{P}}_q(\mbf{\Pi}_{\V^{\perp}}+\mbf{\Pi}_{\V})\mbf{A}_q(x)(\mbf{\Pi}_{\V^{\perp}}+\mbf{\Pi}_{\V}) \\ \notag
     = {\mbf{P}}_q \bigg(\mbf{\Pi}_{\V^{\perp}}\mbf{A}_q(x)  \mbf{\Pi}_{\V^{\perp}} +  \mbf{\Pi}_{\V^{\perp}}\mbf{A}_q(x) \mbf{\Pi}_{\V} \\ \notag  +  \mbf{\Pi}_{\V}\mbf{A}_q(x)\mbf{\Pi}_{\V^{\perp}} +\mbf{\Pi}_{\V}\mbf{A}_q(x)\mbf{\Pi}_{\V} \bigg).
\end{align}
As $\ker(\mbf{P}_q) = \V^\perp$ and thus $ {\mbf{P}}_q\mbf{\Pi}_{\V^{\perp}} = \mbf{0}$, there exists $ {\mbf{P}}_q\mbf{\Pi}_{\V} = {\mbf{P}}_q$ according to \eqref{eq:proj_decom}, which casts \eqref{eq:Aqx} into:
\begin{align}\label{eq:pqaq}
    {\mbf{P}}_q\mbf{A}_q({x})= {\mbf{P}}_q\mbf{A}_q({x})  \mbf{\Pi}_{\V^{\perp}} +  {\mbf{P}}_q \mbf{A}_q({x}) \mbf{\Pi}_{\V}.
\end{align}
Furthermore, since $\V^\perp$ is assumed to be infinitesimally-invariant with respect to the mode $q$, there exists:
\begin{align}\label{eq:invariantuse}
 {\mbf{P}}_q \mbf{A}_q({x}) \mbf{\Pi}_{\V^{\perp}} = {\mbf{P}}_q\mbf{\Pi}_{\V}  \mbf{A}_q({x}) \mbf{\Pi}_{\V^{\perp}} = \mbf{0}
\end{align}
according to \eqref{eq:projection-form}, which further casts  \eqref{eq:pqaq} into:
\begin{align}\label{eq:P_Proj}
    {\mbf{P}}_q\mbf{A}_q({x})= {\mbf{P}}_q\mbf{A}_q({x})  \mbf{\Pi}_{\V}.
\end{align}
The  time derivative $\dot{V} (y(t),t)$ in \eqref{eq:timederivative} thus satisfies:
\begin{align}\label{eq:Lya_der}
    \dot{V} ({y}(t),t) \hspace{-.7mm}= \hspace{-.7mm} {y}^T(t)({\mbf{P}}_q\mbf{A}_q({x}(t))  \mbf{\Pi}_{\V}\hspace{-.7mm} +\hspace{-.7mm}\mbf{\Pi}_{\V}\mbf{A}_q^T({x}(t)){\mbf{P}}_q  ){y}(t)
\end{align}
and due to \eqref{eq:LMI3} and \eqref{eq:LMI4}, there further exist:
\begin{align}\label{eq:stabledecrease}
     \dot{V} ({y}(t),t) \leq -2\eta_{\mc{S}} {V} ({y}(t),t), ~\text{if} ~q \in \mathcal{S}_{\mbf{P}_q}
\end{align}
and:
\begin{align}\label{eq:unstableincrease}
     \dot{V} ({y}(t),t) \leq 2\eta_{\mc{U}} {V} ({y}(t),t), ~\text{if} ~q \in \mathcal{U}_{\mbf{P}_q}
\end{align}
for  $t \in [t_{i-1},t_i)$. Accordingly, the relations:
\begin{align}\label{eq:Lya_bnd_S}
   V\hspace{-.5mm}(y(t^+_i),t^+_i) \hspace{-.8mm}\leq \hspace{-.8mm}\beta_\mc{S} \hspace{-.2mm}\mrm{e}^{-\hspace{-.5mm}2\eta_{\mc{S}}(t_i\hspace{-.5mm}-\hspace{-.5mm}t_{i\hspace{-.5mm}-\hspace{-.5mm}1})} \hspace{-.3mm}   V\hspace{-.5mm}(y(t^+_{i\hspace{-.5mm}-\hspace{-.5mm}1}),t^+_{i\hspace{-.5mm}-\hspace{-.5mm}1}), \text{if} \hspace{.4mm} q \hspace{-.8mm} \in\hspace{-.8mm} \mc{S}_{\mbf{P}_q}
\end{align}
and:
\begin{align}\label{eq:Lya_bnd_U}
   V(y(t^+_i),t^+_i) \hspace{-.5mm}\leq \hspace{-.5mm}\beta_\mc{U} \hspace{-.2mm}\mrm{e}^{2\eta_{\mc{U}}(t_i\hspace{-.5mm}-\hspace{-.5mm}t_{i\hspace{-.5mm}-\hspace{-.5mm}1})}    V\hspace{-.3mm}(y(t^+_{i\hspace{-.5mm}-\hspace{-.5mm}1}),t^+_{i\hspace{-.5mm}-\hspace{-.5mm}1}), \text{if} \hspace{.4mm} q  \hspace{-.8mm}\in \hspace{-.8mm} \mc{U}_{\mbf{P}_q}
\end{align}
must hold after a mode switching has occurred at time $t_i$ according to \eqref{eq:LMI1}. By recursively adopting \eqref{eq:Lya_bnd_S} and \eqref{eq:Lya_bnd_U} for any two adjacent switching times in $\mathcal{T}_\sigma$ starting from $t_0$, it is known that:
\begin{align}\label{eq:Vreduction}
&  V(y(t),t)\leq  \mrm{e}^{\sum\limits_{q\in \mc{S}_{\mbf{P}_q}}N_q(t_0,t)\ln(\beta_\mc{S})- 2\eta_{\mc{S}} T_q(t_0,t)} \cdot \notag  \\ 
  & ~~~~ \qquad   \mrm{e}^{\sum\limits_{ q\in \mc{U}_{\mbf{P}_q}} \hspace{-0.5mm} N_q(t_0,t)\ln(\beta_\mc{U})+2\eta_{\mc{U}} T_q(t_0,t)} V(y(t_0),t_0)
      \end{align}
for all $t \ge t_0$. Then, based on  \eqref{eq:MDADT} and \eqref{eq:MDALT} one knows that:
\begin{align}\label{eq:upperone}
  &  N_q(t_0,t)\ln(\beta_\mc{S})- 2\eta_{\mc{S}} T_q(t_0,t) \notag  \\
  &\leq \underline{N}_{q}\ln(\beta_\mc{S})\hspace{-.8mm} +\hspace{-.8mm} \biggl(\hspace{-.8mm}- 2\eta_{\mc{S}} \hspace{-.8mm}+ \hspace{-.8mm}\frac{{\ln(\beta_\mc{S})}}{\underline{\tau}_{q}}\biggr) T_q(t_0,t), ~\text{if} ~ q \hspace{-.5mm}\in\hspace{-.5mm} \mc{S}_{\mbf{P}_q}
\end{align}
 and similarly:
 \begin{align}\label{eq:uppertwo}
  &  N_q(t_0,t)\ln(\beta_\mc{U}) + 2\eta_{\mc{U}} T_q(t_0,t)\notag  \\ 
  & \leq \overline{N}_{q}\ln(\beta_\mc{U}) \hspace{-.8mm} + \hspace{-.8mm} \biggl(2\eta_{\mc{U}}\hspace{-.8mm} +\hspace{-.8mm} \frac{{\ln(\beta_\mc{U})}}{\overline{\tau}_{q}}\biggr) T_q(t_0,t),  ~\text{if} ~ q \hspace{-.5mm}\in\hspace{-.5mm} \mc{U}_{\mbf{P}_q}.
\end{align}
Due to \eqref{SMADTi} and \eqref{UMADTi}, the term $- 2\eta_{\mc{S}}+\frac{\ln(\beta_\mc{S})}{\underline{\tau}_{q}}$ in \eqref{eq:upperone} and the term $2\eta_{\mc{U}}+\frac{{\ln(\beta_\mc{U}}}{\overline{\tau}_{q}}$ are ensured to be negative. By letting:
 \begin{align}\label{eq:lambdaselection}
\lambda:= \min \{|- 2\eta_{\mc{S}}+\frac{\ln(\beta_\mc{S})}{\underline{\tau}_{q}}|, ~ |2\eta_{\mc{U}}+\frac{{\ln(\beta_\mc{U})}}{\overline{\tau}_{q}}| \}
\end{align}
and defining:
\begin{align}\label{eq:cselection}
     c :=\mrm{e}^{\sum\limits_{q\in \mc{S}_{\mbf{P}_q}} \underline{N}_{q}\ln(\beta_\mc{S})+\sum\limits_{q\in  \mc{U}_{\mbf{P}_q}} \overline{N}_{q}\ln(\beta_\mc{U})},
 \end{align}
the inequality \eqref{eq:Vreduction} is cast into:
 \begin{align}\label{eq:valuereductionfinal}
       V(y(t),t)\leq c \mrm{e}^{-\lambda(t-t_0)} V(y(t_0),t_0).
 \end{align}
 Furthermore, due to \eqref{eq:LMI2} there exists:
\begin{align}\label{eq:reductioninsubspace}
\underline{m}   y^T(t) \mbf{\Pi}_{\V} y(t) \le  V(y(t),t) \le \overline{m}  y^T(t)  \mbf{\Pi}_{\V} y(t)
\end{align}
which yields:
\begin{align}\label{eq:reductioninsubspace2}
\underline{m}     \|\mbf{\Pi}_{\V} y(t)\|_2^2 \le  V(y(t),t) \le \overline{m}     \|\mbf{\Pi}_{\V} y(t)\|_2^2
\end{align}
because of $\mbf{\Pi}_\V^2 = \mbf{\Pi}_\V$. Based on \eqref{eq:reductioninsubspace2}, the inequality \eqref{eq:valuereductionfinal} can be reformulated to:
 \begin{align}
   \|\mbf{\Pi}_{\V} y(t)\|_2^2\leq c \frac{\overline{m}}{\underline{m}}\mrm{e}^{-\lambda(t-t_0)} \|\mbf{\Pi}_{\V} y(t_0)\|_2^2,
 \end{align}
and thus finishes the proof. \hfill$\Box$ 
\end{proof}

The conditions \eqref{eq:LMI1} to \eqref{eq:LMI4} are tailored to the matrix 
$\mbf{P}_q$ in the full state space $\R^n$, whereas the  function $V(y(t),t) = y^T(t) {\mbf{P}}_{\sigma(t)} y(t)$ is defined to measure the distance between $y(t)$ and the origin only within the subspace $\V$. The following demonstrates that these  conditions can be equivalently replaced by those defined on $\V$. First, given the nonsingular and orthogonal matrix $\mbf{T}$ determined by \eqref{eq:Tmatrix}, and defining  $\Tilde{\mbf{A}}_q(x):= \mbf{T}^T \mbf{A}_q (x)\mbf{T} $, it is known that:
\begin{align}\label{eq:Aqtransformation}
\Tilde{\mbf{A}}_q(x)\hspace{-1.2mm} =\hspace{-1.2mm}\begin{bmatrix}
        \mbf{V}^T \hspace{-1mm}\mbf{A}_q(x)\hspace{-.5mm} \mbf{V} &  \mbf{V}^T\hspace{-1mm} \mbf{A}_q(x)\hspace{-.5mm} \mbf{U} \\
         \mbf{U}^T \hspace{-1mm}\mbf{A}_q(x) \hspace{-.5mm} \mbf{V} &  \mbf{U}^T\hspace{-1mm} \mbf{A}_q(x) \hspace{-.5mm}\mbf{U}
    \end{bmatrix}  \hspace{-1.2mm} 
    = \hspace{-1.2mm} \begin{bmatrix}
        \Tilde{\mbf{A}}_{11}^q(x) &  \Tilde{\mbf{A}}_{12}^q(x) \\
          \Tilde{\mbf{A}}_{21}^q(x) &  \Tilde{\mbf{A}}_{22}^q(x)
    \end{bmatrix}
\end{align}
with $\Tilde{\mbf{A}}_{11}^q(x) \in \R^{h\times h}$ and $\Tilde{\mbf{A}}_{2}^q(x) \in \R^{(n-h)\times (n-h)}$. Define also $\Tilde{\mbf{y}}(t) := \mbf{T}^T \mbf{y}(t)$, the time derivative in \eqref{eq:Lya_der} satisfies:
\begin{align}\label{eq:reformulationyt}
  \dot{V} \hspace{-.4mm}(\hspace{-.3mm}{y}(t),\hspace{-.3mm} t\hspace{-.3mm}) &\hspace{-1mm}=\hspace{-1mm} \Tilde{\mbf{y}}\hspace{-.3mm}(t)^T\hspace{-.3mm}\mbf{T}^T \hspace{-.7mm} (\hspace{-.3mm}{\mbf{P}}_q\mbf{A}_q(\hspace{-.4mm}{x}(t)\hspace{-.4mm})  \mbf{\Pi}_{\V}\hspace{-.9mm} +\hspace{-.9mm}\mbf{\Pi}_{\V}\mbf{A}_q^T\hspace{-.4mm}(\hspace{-.4mm}{x}(t)\hspace{-.4mm}){\mbf{P}}_q \hspace{-.3mm} ) \mbf{T}\Tilde{\mbf{y}}\hspace{-.3mm}(t)\notag \\
 & = \Tilde{\mbf{y}}(t)^T \mbf{T}^T {\mbf{P}}_q  \mbf{T} \Tilde{\mbf{A}}_q(x(t)) \mbf{T}^T  \mbf{\Pi}_{\V} \mbf{T}\Tilde{\mbf{y}}(t) \notag \\
 &   \qquad +  \Tilde{\mbf{y}}(t)^T \mbf{T}^T \mbf{\Pi}_{\V}\mbf{T}\Tilde{\mbf{A}}^T_q(x(t))\mbf{T}{\mbf{P}}_q \mbf{T}\Tilde{\mbf{y}}(t).
\end{align}
Based on \eqref{eq:P_tilde}, there exists:
 \begin{align}\label{eq:P_tildenewq}
\mbf{T}^T \mbf{P}_q \mbf{T} = \begin{bmatrix}
    \Tilde{\mbf{P}}_q & \mbf{0}\\
    \mbf{0}& \mbf{0}
\end{bmatrix},~\Tilde{\mbf{P}}_q \in \R^{h\times h},
 \end{align}
and according to \eqref{eq:proj_decom}, the time derivative in \eqref{eq:reformulationyt} can be further reformulated into:
\begin{align}\label{eq:Lya_sim}
     \dot{V}\hspace{-.5mm}  (\mbf{y}(t), t)\hspace{-.7mm}  = \hspace{-.7mm} 
\Tilde{\mbf{y}}(t)^T\hspace{-.5mm} (\hspace{-.5mm} \begin{bmatrix}
    \Tilde{\mbf{P}}_q\Tilde{\mbf{A}}_{11}^q({x}(t))\hspace{-1mm} +\hspace{-1mm}\Tilde{\mbf{A}}_{11}^q(x(t))^T  \Tilde{\mbf{P}}_q&\mbf{0}\\
    \mbf{0} & \mbf{0}
\end{bmatrix}\hspace{-.5mm} )  \Tilde{\mbf{y}}(t).
\end{align}
Moreover, since there exist:
\begin{align}\label{eq:l2-measurenew}
 \mu_{2,\mbf{P}^{1/2}}(\Tilde{\mbf{A}}_q)=&\argmin\limits_{\mathbf{P}\mbf{T}^T \mbf{A}_q\mbf{T}\boldsymbol{\Pi}_{\V^\perp}+ \boldsymbol{\Pi}_{\V^\perp}\mbf{T}\mbf{A}^T_q  \mbf{T}^T \mathbf{P} \preceq 2b \mathbf{P} }b \notag \\
&  =\argmin\limits_{\Tilde{\mbf{P}}_q\Tilde{\mbf{A}}_{11}^q + \Tilde{\mbf{A}}_{11}^{q,T}\Tilde{\mbf{P}}_q \preceq 2b \Tilde{\mbf{P}}_q }b
\end{align}
due to \eqref{eq:proj_decom} and  \eqref{eq:P_tildenewq}, a mode $q \in \mathcal{M}$ that belongs to the set $\mc{S}_{\mbf{P}_q}$ according to \eqref{eq:stablemodes}, it must also belong to the set $\mc{S}_{\Tilde{\mbf{P}}_q }$ (same for  $\mc{U}_{\Tilde{\mbf{P}}_q }$). Finally, the   conditions  \eqref{eq:LMI1} - \eqref{eq:LMI4}  are noticed to be equivalent to:
\begin{align}
   &  {\Tilde{\mbf{P}}}_{\sigma(t^+_i)} \preceq \beta_\mc{S} {\Tilde{\mbf{P}}}_{\sigma(t_i^-)}, ~ \text{if} ~ \sigma(t_i^-) \in \mc{S}_{{\Tilde{\mbf{P}}}_{\sigma(t_i^-)}} \label{eq:LMI1new}\\
   &  {\Tilde{\mbf{P}}}_{\sigma(t^+_i)} \preceq \beta_\mc{U} {\Tilde{\mbf{P}}}_{\sigma(t_i^-)}, ~ \text{if} ~ \sigma(t_i^-) \in \mc{U}_{{\Tilde{\mbf{P}}}_{\sigma(t_i^-)}} \label{eq:LMI12new}\\
 &   \underline{m} \mbf{I}_h \preceq \Tilde{\mbf{P}}_q \preceq  \overline{m} \mbf{I}_h, ~\forall q \in  \mathcal{M} \label{eq:LMI2new} \\
&  \Tilde{\mbf{P}}_q\Tilde{\mbf{A}}_{11}^q({x})\hspace{-1mm} +\hspace{-1mm}\Tilde{\mbf{A}}_{11}^q(x)^T\Tilde{\mbf{P}}_q \preceq \hspace{-0.5mm}-2\eta_{\mc{S}}     \Tilde{\mbf{P}}_q,~ \text{if} ~ q \in \mc{S}_{\Tilde{\mbf{P}}_q } \label{eq:LMI3new} \\
&  \Tilde{\mbf{P}}_q\Tilde{\mbf{A}}_{11}^q({x})\hspace{-1mm} +\hspace{-1mm}\Tilde{\mbf{A}}_{11}^q(x)^T\Tilde{\mbf{P}}_q \hspace{-0.5mm}\preceq \hspace{-0.5mm}2\eta_{\mc{U}}     \Tilde{\mbf{P}}_q, ~\text{if} ~ q \in \mc{U}_{\Tilde{\mbf{P}}_q }. \label{eq:LMI4new}
\end{align}
In cases where   $h \ll n$, it is only necessary to search for a matrix  $\Tilde{\mbf{P}}_q \in \R^{h\times h}$  with a much lower dimension than $\mbf{P}_q$ to fulfill \eqref{eq:LMI1new} - \eqref{eq:LMI4new}.  Subsequently, the  matrix $\mbf{P}_q$ that satisfies \eqref{eq:LMI1} - \eqref{eq:LMI4} can be directly recovered according to \eqref{eq:P_tildenewq}.

\subsection{Contraction Analysis for the State Space $\R^n$}

Note that Lemma~\ref{lem:sem-con} establishes conditions under which the state $y(t)$ of the variational system  \eqref{eq:var_sys} is exponentially stable in a subspace $\V$. To ensure that $y(t)$ is also stable in  $\R^n$, a  family of seminorms  $\mc{F} = \{|||\cdot |||_i\}_{i \in \mc{I}}$, $\mc{I} := \{1,2,...,I\}$, that constitutes  a separating family on $\R^n$ according to \eqref{eq:conditionseperating} is considered.


\begin{theorem}\label{theorem:sem-conyt}

For a set of subspace $\V_i \subseteq \R^n$,  $i \in \mc{I}$, it is assumed that  the semi-norms     $|||\cdot|||_i=\| \mbf{\Pi}_{\V_i} \,(\cdot)\,\|_2$ with $\mrm{ker}(|||\cdot |||_i)=\V^\perp_i$ form a separating family $\mc{F} = \{|||\cdot |||_i\}_{i \in \mc{I}}$  of $\R^n$. Furthermore, each $\V^\perp_i$ is infinitesimally-invariant with respect to all modes   $q \in \mathcal{M}$ of   \eqref{eq:sys}. Then, given a switching  signal $\sigma(t)$, suppose that the conditions \eqref{eq:LMI1} -  \eqref{eq:LMI4} are satisfied for each $\V_i$,  $i \in \mc{I}$, by:
\begin{itemize}
\item  a set of semi-positive definite matrices $\mbf{P}_{q,i}\in \R^{n\times n}$, $\ker(\mbf{P}_{q,i}) = \V^\perp_i$, for each mode $q \in  \mathcal{M}$, and;
\item  a set of constants $\beta_{\mc{S}, i} >1$ and $\beta_{\mc{U}, i} \in (0,1)$ in \eqref{eq:LMI1}, and $\underline{m}_i, \overline{m}_i, \eta_{\mc{U}, i}, \eta_{\mc{S},i} >0$ in \eqref{eq:LMI2} - \eqref{eq:LMI4}.
\end{itemize}
If the average dwell and  leave  time of each mode $q\in  \mathcal{M}$ satisfy:
 \begin{align} 
    \underline{\tau}_{q} &> \frac{ \max\limits_{i \in \mc{I}}  \ln \beta_{\mc{S}, i}}{2 \cdot \min\limits_{i \in \mc{I}} \eta_{\mc{S}, i}}, ~  \text{if}~ q \in  \mc{S}_{\mbf{P}_{q,i}},~ \forall i \in \mc{I}\label{newSMADTi}\\
    \overline{\tau}_{q} &<-\frac{\max\limits_{i \in \mc{I}}  \ln \beta_{\mc{U}, i} }{2 \cdot \min\limits_{i \in \mc{I}} \eta_{\mc{U}, i}}, ~ \text{if}~ q \in \mc{U}_{\mbf{P}_{q,i}},~ \forall i \in \mc{I}\label{newUMADTi}
\end{align}
in the switching signal $\sigma(t)$, then   the  switched system \eqref{eq:sys} is uniformly contracting.
\end{theorem}

\begin{proof}
For $\underline{\tau}_{q}$ and $\overline{\tau}_{q}$ satisfying \eqref{newSMADTi} and \eqref{newUMADTi}, they must also satisfy  \eqref{SMADTi} and \eqref{UMADTi} for each subspace $\V_i$, $i \in \mc{I}$, which implies the existence of constants $k_i, \lambda_i>0$ satisfying:
\begin{align}
    \|\mbf{\Pi}_{\V_i} y(t)\|_2\leq k_i \mrm{e}^{-\lambda_i(t-t_0)} \|\mbf{\Pi}_{\V_i} y(t_0)\|_2 
\end{align}
according to  Lemma~\ref{lem:sem-con}. Then,   as   a  norm $\| \cdot \|_{\mc{F}}$ on $\R^n$ can be induced by  the separating family $\mc{F}$, there must exist:
    \begin{align}\label{eq:stabilizingytall}
        \|\mbf{y}(t)\|_{\mc{F}} &\leq  k^{*} \mrm{e}^{-\lambda^{*}(t-t_0)} \|\mbf{y}(t_0)\|_{\mc{F}},
    \end{align}
with $k^* >0$ and $\lambda^* := \min_{i \in \mathcal{I}}(\lambda_i)$. As a result, the variational system \eqref{eq:var_sys} is exponentially stabilizing in $\R^n$ and thus the switched system \eqref{eq:sys} is contracting according to Theorem.~\ref{theorem:contractingvariation}. \hfill$\Box$ 
\end{proof}
\begin{remark}
For a given switching signal $\sigma(t)$, let $\hat{T}_{q,r,\sigma}  \in (0, \infty)$, $r \in \mathbb{N}$, denote the activation time of mode $q \in \mathcal{M}$
between the $r$-th and $r+1$-th activations in  $\sigma(t)$. Based on the conditions  \eqref{newSMADTi} and \eqref{newUMADTi} for the average dwell and  leave  times, it is known that when:
 \begin{align*} 
& \hat{T}_{q,r,\sigma}  \ge \frac{ \max\limits_{i \in \mc{I}}  \ln \beta_{\mc{S}, i}}{2 \cdot \min\limits_{i \in \mc{I}} \eta_{\mc{S}, i}}, ~  \text{if}~ q \in  \mc{S}_{\mbf{P}_{q,i}},~ \forall i \in \mc{I}\\
& \hat{T}_{q,r,\sigma}  \le -\frac{\max\limits_{i \in \mc{I}}  \ln \beta_{\mc{U}, i} }{2 \cdot \min\limits_{i \in \mc{I}} \eta_{\mc{U}, i}}, ~ \text{if}~ q \in \mc{U}_{\mbf{P}_{q,i}},~ \forall i \in \mc{I}
\end{align*}
hold for all  $r \in \mathbb{N}$ and for all modes $q \in \mathcal{M}$ in  $\sigma(t)$, then the switched system \eqref{eq:sys} must be contracting according to Theorem.~\ref{theorem:sem-conyt}. Note that although the conditions on the  activation time  $\hat{T}_{q,r,\sigma}$ are more conservative than those for the average time ones, they provide an easier way to determine whether a given switching signal leads to contracting behavior.
\end{remark}


\section{Numerical Examples}
\label{sec:example}
Consider the switched system \eqref{eq:sys} with nonlinear dynamics:
\begin{align*}
&  q=1:\\
&    \begin{bmatrix}
        \dot{x}_1\\
        \dot{x}_2
    \end{bmatrix}\hspace{-1mm} = \hspace{-1mm} \begin{bmatrix}
        -\frac{3}{4}& -\frac{5}{4}\\
        -\frac{5}{4}&  -\frac{3}{4}
    \end{bmatrix} \begin{bmatrix}
        x_1\\
        x_2
    \end{bmatrix} \hspace{-1mm}   + \hspace{-1mm} \begin{bmatrix}
        -\frac{0.2}{\sqrt{2}}&- \frac{0.7}{\sqrt{2}} \\
         -\frac{0.2}{\sqrt{2}}& \frac{0.7}{\sqrt{2}}
    \end{bmatrix} \begin{bmatrix}
        \cos(\frac{0.1}{\sqrt{2}}(x_1 \hspace{-.7mm} +\hspace{-.7mm} x_2))\\
        \cos(\frac{0.1}{\sqrt{2}}(x_1\hspace{-.7mm} -\hspace{-.7mm} x_2))
    \end{bmatrix} \\
  & q=2: \\
  &  \begin{bmatrix}
        \dot{x}_1\\
        \dot{x}_2
    \end{bmatrix}  \hspace{-.7mm}= \hspace{-.7mm} \begin{bmatrix}
        -\frac{3}{4}& \frac{5}{4}\\
        \frac{5}{4}&  -\frac{3}{4}
    \end{bmatrix} \begin{bmatrix}
        x_1\\
        x_2
    \end{bmatrix}  \hspace{-.7mm}  +  \hspace{-.7mm}\begin{bmatrix}
        \frac{0.7}{\sqrt{2}}& \frac{0.2}{\sqrt{2}} \\
         \frac{0.7}{\sqrt{2}}&- \frac{0.2}{\sqrt{2}}
    \end{bmatrix}\begin{bmatrix}
        \cos(\frac{0.1}{\sqrt{2}}(x_1 \hspace{-.7mm}+ \hspace{-.7mm}x_2))\\
        \cos(\frac{0.1}{\sqrt{2}}(x_1 \hspace{-.7mm}- \hspace{-.7mm}x_2))
    \end{bmatrix}.
\end{align*}
\begin{figure}[!t]
    \centering
    \includegraphics[width=.7\linewidth]{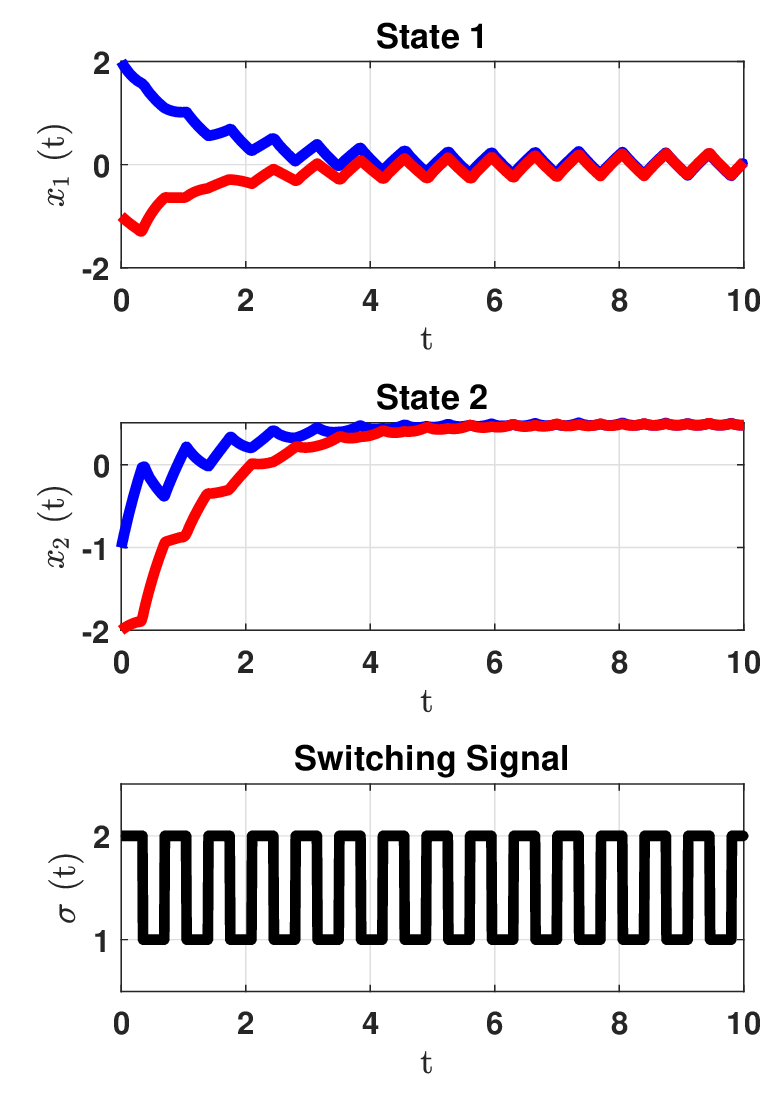}
    \caption{Contracting evolution starting from two different initial states under a periodic switching signal $\sigma(t)$.}
    \label{fig:Sim1}
\end{figure}
Both modes can be identified to be non-contracting as the variational systems  \eqref{eq:var_sys} are unstable. Select the subspaces:
\begin{align*}
  &  \V_1 = \mrm{span}\biggl\{\frac{1}{\sqrt{2}}\begin{bmatrix}
        1\\
        1
    \end{bmatrix}\biggr\},  ~\mbf{\Pi}_{\V_1} = \frac{1}{2} \begin{bmatrix}
        1 &-1\\
        -1& 1
    \end{bmatrix}\\
 &    \V_2 = \mrm{span}\biggl\{\frac{1}{\sqrt{2}}\begin{bmatrix}
        1\\
        -1
    \end{bmatrix}\biggr\}, ~  \mbf{\Pi}_{\V_2} = \frac{1}{2} \begin{bmatrix}
        1 &1\\
        1& 1
    \end{bmatrix}
\end{align*}
and define  $|||\cdot|||_i=\| \mbf{\Pi}_{\V_i} \,(\cdot)\,\|_2$ with $\mrm{ker}(|||\cdot |||_i)=\V^\perp_i$, it is known that $ \ker(|||\cdot|||_1) \cap \ker(|||\cdot|||_2) = \{\mbf{0}\}$, i.e., the two seminorms constitute a separating family $\mc{F}$  of $\R^2$. Then, for semi-positive matrices:
\begin{align*}
  &q\hspace{-1mm} =\hspace{-1mm} 1:~ \mbf{P}_{1,1}\hspace{-1.35mm}  =\hspace{-1.5mm}  \begin{bmatrix}
       1.1389 &\hspace{-1mm}   -\hspace{-.4mm}1.1389\\
   -\hspace{-.4mm}1.1389  & \hspace{-1mm}  1.1389
   \end{bmatrix} ,  \mbf{P}_{1,2}\hspace{-1.35mm} = \hspace{-1.5mm} \begin{bmatrix}
      0.7081 &  \hspace{-1mm} -\hspace{-.2mm}0.7081\\
   -\hspace{-.2mm}0.7081   &\hspace{-1mm}  0.7081
   \end{bmatrix}\\
& q\hspace{-1mm} =\hspace{-1mm} 2: ~ \mbf{P}_{2,1}\hspace{-1mm}  = \hspace{-1mm}\begin{bmatrix}
       0.7081  &\hspace{-.5mm}  0.7081\\
    0.7081  & \hspace{-.5mm} 0.7081
   \end{bmatrix},\mbf{P}_{2,2} \hspace{-1mm}= \hspace{-1mm}\begin{bmatrix}
    1.1389  & \hspace{-.5mm} 1.1389\\
    1.1389  & \hspace{-.5mm} 1.1389
   \end{bmatrix}
\end{align*}
the conditions \eqref{eq:LMI1} -  \eqref{eq:LMI4} are satisfied by modes with constants: 
\begin{align*}
&\beta_{\mc{S},1} = \beta_{\mc{S},2}  = 1.6084, ~\beta_{\mc{U},1}  = \beta_{\mc{U},2}  = 0.6217, \\
&\eta_{\mc{S},1} = \eta_{\mc{S},2} = -1.5, ~\eta_{\mc{U},1} = \eta_{\mc{U},2} = 0.6.
\end{align*}
According to Theorem~\ref{theorem:sem-conyt}, the system is contracting when the average dwell and leave times satisfy:
\begin{align*}
 \underline{\tau}_{1}>0.1584, ~ \underline{\tau}_{2}>0.1584,~  \overline{\tau}_{1}<0.3960, ~\overline{\tau}_{2}<0.3960.
\end{align*}
Thus, for a periodically switching signal $\sigma(t)$ with  $\hat{T}_{1,r,\sigma}=\hat{T}_{2,r,\sigma} = 0.35$, the state trajectories of \eqref{eq:sys} from different initial states  approach to each other, as shown in Fig.~\ref{fig:Sim1}. In another test, a non-periodic switching signal that satisfies the average dwell and leave times is considered, see Fig.~\ref{fig:Sim2}. It is shown that the trajectories are still  contracting as expected.
\begin{figure}[!t]
    \centering
    \includegraphics[width=.7\linewidth]{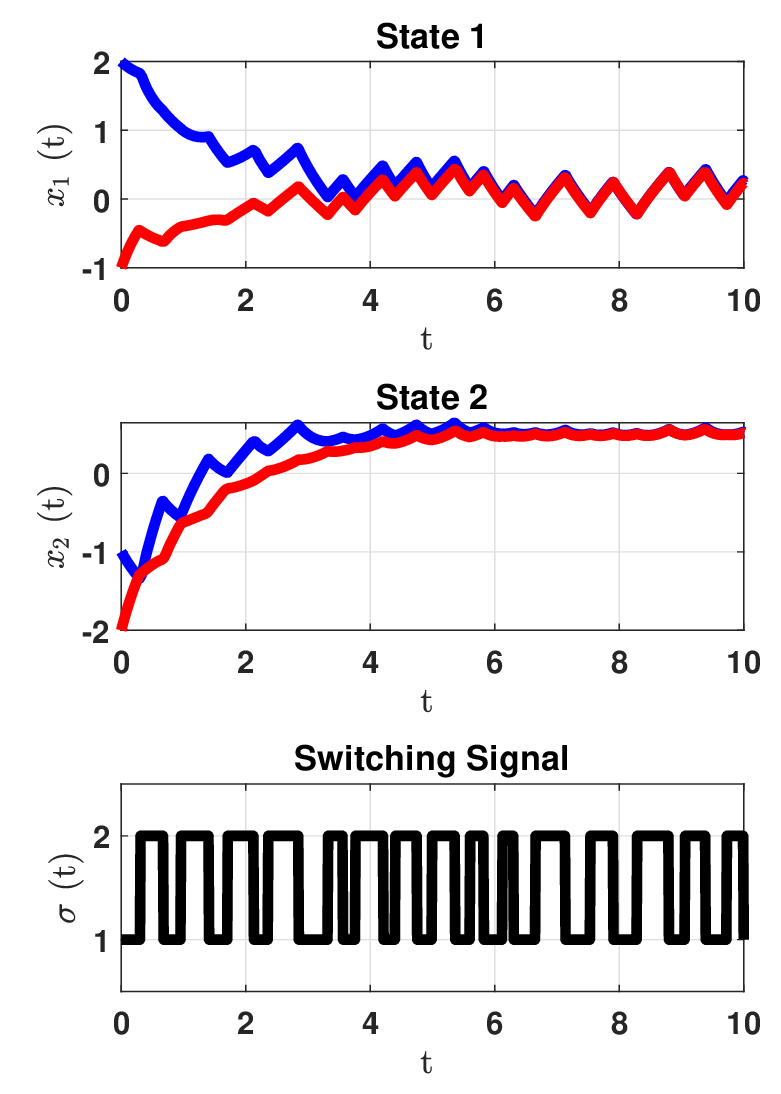}
    \caption{Contracting evolutions for a non-periodic switching signal $\sigma(t)$.}
    \label{fig:Sim2}
\end{figure}

\section{Conclusion}
\label{sec:conclusion}
In this paper, the contraction properties of switched systems composed solely of non-contracting modes are studied. Unlike existing literature, which requires at least one mode to be contracting, it is demonstrated here that by decomposing the state space into several subspaces and satisfying certain switching time conditions related to the state evolution within those subspaces, the switched systems can still exhibit contracting behavior. A crucial step to derive these conditions is the selection of the subspaces, which must form a separating family and be infinitesimally invariant for each mode. Since this selection criterion can be quite restrictive, future work will focus on developing less conservative and more efficient methods for determining the subspaces.

\bibliography{IFACWC_Bibliography}             

\end{document}